\documentclass[a4paper,reqno]{amsart}
\usepackage[english]{babel}
\usepackage[T1]{fontenc}
\usepackage{amsmath,amssymb,amsfonts,amsthm}
\usepackage{hyperref}
\usepackage{slashed}
\usepackage{graphicx,color}
\usepackage[bbgreekl]{mathbbol}
\usepackage{mathtools}
\usepackage[all,cmtip]{xy}

\newtheorem{definition}{Definition}
\newtheorem{theorem}[definition]{Theorem}
\newtheorem{proposition}[definition]{Proposition}
\newtheorem{lemma}[definition]{Lemma}

\newtheorem{remark}[definition]{Remark}

\newcommand{\qsp}[2]{\,\ensuremath{\raise.5ex\hbox{$#1$}\big\slash\raise-.5ex\hbox{$#2$}}} 

\newcommand{\pard}[2]{\frac{\delta#1}{\delta#2}}

\newcommand{\intl}{\int\limits}

\newcommand{\Wedge}[1]{{\textstyle \bigwedge^{#1}}}

\title{BV-equivalence between triadic gravity and BF theory in three dimensions}
\author{A. S. Cattaneo}
\address{institut f\"ur Mathematik, Winterthurerstrasse 190, 8057 Z\"urich, Switzerland}
\email{cattaneo@math.uzh.ch}

\author{M. Schiavina}
\address{Department of Mathematics, University of California, Berkeley, 920 Evans Hall, Berkeley CA 94720-3840 USA}
\email{michele.schiavina@berkeley.edu}

\author{I. Selliah}
\address{institut f\"ur Mathematik, Winterthurerstrasse 190, 8057 Z\"urich, Switzerland}
\email{iswaryaa.selliah@uzh.ch}

\thanks{This research was (partly) supported by the NCCR SwissMAP, funded by the Swiss National Science Foundation and by the COST Action MP1405 QSPACE, supported by COST (European Cooperation in Science and Technology).
A. S. C. acknowledges partial support of SNF grant No. 200020\_172498$\slash$1. M. S. is supported by SNF grant No. P2ZHP2\_164999. }

\begin{document}

\maketitle

\begin{abstract}
The triadic description of General Relativity in three dimensions is known to be a BF theory. Diffeomorphisms, as symmetries, are easily recovered on shell from the symmetries of BF theory. This note describes an explicit off-shell BV symplectomorphism between the BV versions of the two theories, each endowed with their natural symmetries.
\end{abstract}

\section*{Introduction}
A $d$-dimensional (classical) field theory is the (local) assignment of a \emph{space of fields} ${F}_M$ and an \emph{action functional} $S^0_M\colon {F}_M \longrightarrow \mathbb{R}$ to any $d$-dimensional manifold $M$ (possibly with boundary), and most of its relevant physical information is contained in the critical locus of the action $EL_M$. 

Under this perspective, any two assignments of a space of fields and a functional that have diffeomorphic critical loci are candidates to be considered \emph{equivalent} from a physical standpoint.

However, when the theory admits an additional piece of data, a symmetry distribution $D_M$ in $TF_M$, such that any of its sections annihilates $S_M$, the picture becomes more involved. Since configurations that are related by a symmetry transformation are physically undistinguishable, the true space one looks at is the moduli space of solutions up to said equivalence: $\qsp{EL}{D_M}$. Such space is generically \emph{bad} - e.g. a stack of some sort - and one goes about the algebraic geometry of its description resolving the quotient by means of the Chevalley-Eilenberg-Koszul-Tate complex. 

This construction goes under the name of BV formalism and outputs the dg-Poisson-algebra of functions over a (-1)-symplectic graded manifold, seen as an \emph{extended} space of fields $\mathcal{F}$ (in the sense that its degree-zero part is the original space of fields), and a vector field $Q$ of degree 1 (the differential of the dg-algebra) encoding the symmetries. In addition, $Q$ is the Hamiltonian vector field for a function $S_M$, the \emph{extended} BV action\footnote{Extended in the same sense used before.} \cite{BV81, Stash,stash97}.

In this dg-setting one can easily make sense of a notion of an equivalence of BV data obtained, say, as alternative extension of the same classical data, or as completely independent constructions, and since the cohomology in degree zero of the (classical) BV complex is, by construction, a presentation of the algebra of functions over $\qsp{EL_M}{D_M}$, comparison of the cohomologies will provide a natural language to understand equivalence of theories.

There are situations, though, where something stronger can be said about two theories, when they are equivalent in a much stricter sense and an actual BV-symplectomorphism between two BV complexes (viewed as functions on graded symplectic manifolds) can be found, such that the action functionals are pulled back to one another.

In this paper we will show that this is indeed the case for General Relativity in 3 dimensions in the triadic formalism and nonabelian BF theory, where an explicit canonical transformation will be given in terms of generating functions. As a matter of fact, it is well known that the triadic version of 3D gravity is classically a BF theory. The geometric symmetries for GR contain diffeomorphisms. On shell they can be recovered from the symmetries of BF theory \cite{Witten}. We extend this result off-shell in the BV formalism. 

The formulation of GR we will be using here is often (ambiguosly) referred to it as Palatini--Cartan (PC) theory \cite{Palatini,Cartan}. Its analysis in 4 dimensions is carried out in \cite{CS4}, where the ambiguity in the nomenclature is also explained. For more details on gravity in 2+1 dimensions, see \cite{Carl} and references therein.

\section{Preliminaries}
In this section we will review the basics of the BV formalism, the triad Palatini--Cartan formulation of 3d - General Relativity as a gauge theory, of its counterpart BF theory and how both can be extended to yield BV data.

\subsection{BV formalism}
We are interested in extending a classical field theory in the cohomological framework developed by Batalin and Vilkovisky \cite{BV81}. The construction is aimed at replacing the reduced critical locus (solutions of the Euler-Lagrange equations modulo gauge equivalence) with a smooth chain complex, seen as a graded manifold. The gauge fixing procedure is interpreted as a choice of a suitable Lagrangian submanifold of this extended space; however, in the present paper we shall not be concerned with quantisation. 

Here are the main definitions we will need.
\begin{definition}
A \emph{BV manifold} is a 4-tuple $(\mathcal{F},\Omega, S, Q)$, consisting of a $(-1)$-symplectic graded manifold $(\mathcal{F},\Omega)$, a \emph{cohomological} vector field $Q\in\mathfrak{X}[1](\mathcal{F})$, i.e. such that $[Q,Q]=0$ and a \emph{BV action}, a function $S\colon \mathcal{F}\longrightarrow \mathbb{R}$ such that $\iota_Q\Omega = \delta S$.
\end{definition}

The main example shows up in field theory, where we may construct a BV manifold for each input spacetime $M$, as follows:
\begin{definition}
A \emph{$d$-dimensional BV theory} is the assignment to every closed $d$-dimensional manifold $M$ of a BV manifold $(\mathcal{F}_M,\Omega_M,Q_M,S_M)$, given in terms of local data.
\end{definition}

\begin{remark}
Observe that when a boundary is present the equation for the Hamiltonian function $\iota_Q\Omega = \delta S$ is likely to be spoiled by a boundary term. This is handled in the BV-BFV formalism of Cattaneo, Mn\"ev and Reshetikhin \cite{CMR1}. However, because the whole construction is local, it still makes sense to consider the BV-theory as if there were no boundary.
\end{remark}

Usually, the BV theory is constructed as an extension of some classical field theory with gauge symmetry, i.e. the assignement of some space of fields $F_M$ a local action functional $S_M^0$ on it and a symmetry distribution $D_M$, provided that the associated distribution is involutive on-shell. Most of the examples that show up in physics are such that the symmetry distribution is given by the action of a Lie algebra on the space of classical fields, e.g. Yang-Mills theory, Chern-Simons theory, but also BF theory and General Relativity \cite{CMR1,CS1,CS2}. We call these theories BRST-like, after the BRST construction \cite{BRST}, and the associated BV-extension is \emph{minimal} in some sense as determined by \cite{BV81}:
\begin{theorem}\label{minimalBV}
Let $(F_M, S_M^0,D_M)$ define a classical field theory with gauge symmetry. If the distribution $D_M$ comes from a Lie algebra action, the functional 
$$S_M=S^0_M + \langle \Phi^\dag,Q\Phi\rangle$$ 
on the space of fields $\mathcal{F}_M=T^*[-1]D_M[1]$, where $\Phi$ is a multiplet of fields in $D_M[1]$ and $\Phi^\dag$ denotes the corresponding multiplet of dual fields, yields a BV theory together with $Q$, the degree $1$ vector field encoding the symmetries of $D_M$, .
\end{theorem}

The definition we will need to state the main result of this work is the following:
\begin{definition}
A \emph{strong equivalence} between the BV-theories 
$$\mathfrak{F}_M^{(1|2)}\coloneqq \left(\mathcal{F}_M^{(1|2)},\Omega_M^{(1|2)},Q_M^{(1|2)},S_M^{(1|2)}\right) $$
is a graded symplectomorphism $\Phi\colon \left(\mathcal{F}_M^{(1)},\Omega_M^{(1)}\right) \longrightarrow \left(\mathcal{F}_M^{(2)},\Omega_M^{(2)}\right)$ preserving the BV-action, i.e. $\Phi^*S_{M}^{(2)}=S_M^{(1)}$.
\end{definition}

\subsection{3d General Relativity}
Consider $P\longrightarrow M$ an $SO(2,1)$ bundle over an orientable manifold $M$ (for simplicity we will assume $\partial M=0$) and let $\mathcal{W}\longrightarrow M$ be an associated vector bundle endowed with a smooth fiberwise Minkowski metric $(W,\eta)$ and with an orientation. A co-frame field, sometimes called a triad or a \emph{dreibein}, is a bundle isomorphism $e\colon TM\longrightarrow \mathcal{W}$ covering the identity. We will thus consider $e\in\Omega_{nd}^1(M,\mathcal{W})$ (the subscript \emph{nd} stands for nondegenerate) and use the isomorphism $\mathfrak{so}(2,1)\simeq \Wedge{2}\mathcal{W}\simeq\mathcal{W}$ induced by the metric and the internal hodge dual. Given a connection $\omega\in\mathcal{A}_P$ on $P$, its curvature $F_\omega$ will be regarded as a $\Wedge{2}$-valued two-form one-form.

Denote by ${F}_{GR}\coloneqq \Omega_{nd}^1(M,\mathcal{W})\times \mathcal {A}_P$ the space of physical fields and consider the action functional
\begin{equation}
S^0_{GR}(\Lambda)=\intl_{M} \mathrm{Tr} [e\wedge F_\omega + \frac{\Lambda}{3} e^3]
\end{equation}
with the trace denoting the pairing with volume form in $\Wedge{3}W$ and $\Lambda\in \mathbb{R}$ the cosmological constant. The Euler--Lagrange equations are given by
\begin{subequations}\begin{align}\label{Einstein}
F_\omega = 0 \\\label{halfshell}
d_\omega e = 0
\end{align}\end{subequations}
and it is a well-known result that solving \eqref{halfshell} to yield $\omega=\omega(e)$ one obtains the Levi--Civita connection for the metric $g=e^*\eta$ and \eqref{Einstein} then reduces to the Einstein equations for $g$. Observe that with this redefinition, $S^0_{GR}$ reduces to the standard Einstein--Hilbert action functional $S_{EH}=\int_M \sqrt{-|g|}R[g]$.

The symmetries of this theory are given by  the action of (infinitesimal) diffeomorfisms, and internal $SO(2,1)$ gauge transformations. Adapting to three dimensions the costruction of \cite{CS2} we extend the classical data to a BV manifold by declaring the following
\begin{align}
Q_{GR}(e) = &  L^\omega_\xi e + [\chi,e] \\
Q_{GR}(\omega) = & \iota_\xi F_\omega + d_\omega \chi\\
Q_{GR}(\xi) = & L_\xi\xi \\ 
Q_{GR}(\chi) = & \frac12( [\chi,\chi] - \iota_\xi \iota_\xi F_\omega).
\end{align} 
where $L_\xi^A\coloneqq [\iota_\xi,d_A]$. This defines the cohomological vector field $Q_{GR}$ of degree 1, with $\xi\in\mathfrak{X}[1](M)$ and $\chi\in\Omega(M,\mathrm{ad}P)$ the \emph{ghost} fields, the space of BV-fields is 
\begin{equation}
\mathcal{F}_{GR} = T^*[-1]\left(F_{GR}\times \mathfrak{X}[1](M) \times \Omega(M,\mathrm{ad}P)\right)
\end{equation}
If we decorate cotangent fields with a dagger, we can easily verify that the corresponding BV-action is
\begin{multline}
S_{GR}(\Lambda) = S^0_{GR}(\Lambda) + \intl_{M} \langle e^\dag, Q_{GR}(e)\rangle + \langle A^\dag, Q_{GR}(A)\rangle +\langle \chi^\dag, Q_{GR}(\chi)\rangle + \langle \xi^\dag, Q_{GR}(\xi)\rangle=\\
=\intl_{M} \mathrm{Tr} [e\wedge F_\omega + e^\dag (L^\omega_\xi e + [\chi,e]) + A^\dag(\iota_\xi F_\omega + d_\omega \chi) + \frac12\chi^\dag( [\chi,\chi] - \iota_\xi \iota_\xi F_\omega ) + \frac12\iota_{[\xi,\xi]}\xi^\dag]
\end{multline}

It follows immediately that the 4-tuple $\mathfrak{F}_{GR}\coloneqq(\mathcal{F}_{GR},\Omega_{GR}, Q_{GR},S_{GR})$ defines a BV-theory \cite{CS2}.

\subsection{BF theory}
Let us assume again that we are given the $SO(2,1)$ principal bundle $P\longrightarrow M$ and the associated oriented \emph{Minkowski bundle} $\mathcal{W}\longrightarrow M$. We want to define BF theory as a classical field theory, so we consider the space of fields $F_{BF}\coloneqq \Omega^1(M,\Wedge{2}\mathcal{W}^*)\times \mathcal{A}_P\ni(B,A)$ together with the action functional
\begin{equation}
S_{BF}^0\coloneqq \intl_{M} \langle B, F_A\rangle \equiv \intl_{M} \mathrm{Tr}[B F_A]
\end{equation}
where we identify $\Wedge{2}\mathcal{W}^*$ with $\mathcal{W}$ and use the volume form in $\Wedge{3}W$.

The symmetries of this action are given by two sets of transformations on the fields, one accounting for the internal $SO(2,1)$ gauge symmetry, while the other stems from the fact that $B$ can be shifted by a covariantly-exact form. In other words, we can construct the degree-1 vector field
\begin{align}
Q_{BF}(B) = &  d_A\tau + [c,B] \\
Q_{BF}(A) = &  d_A c\\
Q_{BF}(\tau) = & [c,\tau] \\ 
Q_{BF}(c) = &  \frac12 [c,c]
\end{align} 
over the space of BV fields
\begin{equation}
\mathcal{F}_{BF}\coloneqq T^*[-1]\left(F_{BF}\times \Omega^0(M,\mathcal{W})\times \Omega^0(M,\Wedge{2}\mathcal{W})\right)
\end{equation}
with $\tau\in\Omega^0(M,\mathcal{W}), c\in\Omega^0(M,\Wedge{2}\mathcal{W})$, and the BV extended action then reads
\begin{equation}
S_{BF}=\intl_{M} \mathrm{Tr}\left[BF_A + B^\dag(d_A\tau + [c,B]) + A^\dag d_Ac + \tau^\dag[c,\tau] + \frac12 c^\dag[c,c] \right]
\end{equation}

Summarising, it is a known result that the 4-tuple $\mathfrak{F}_{BF}\coloneqq(\mathcal{F}_{BF},\Omega_{BF},Q_{BF},S_{BF})$ defines a BV-theory.

\begin{remark}[Classical equivalence]
The first obvious observation that one could make at this point is that when looking at the degree-zero part of the BV-manifold, i.e. at the non-extended theory, there is an obvious map between the open submanifold of ${F}_{BF}$ consisting of non degenerate vector valued one forms and connections, to $F_{GR}$. That map is the identity map. 

One can also observe that the two sets of symmetries coincide on shell \cite{Witten}, as is easily verified as follows: choose $\tau=-\iota_\xi B$ and $c=-\iota_\xi A$, then 
$$d_A\tau = d_A\iota_\xi B =L_\xi^A B - \iota_\xi d_A B\approx L_\xi B$$
together with
$$d_Ac=-d_A\iota_\xi A = L_\xi^A A - \iota_\xi F_A \approx L_\xi A$$
where the symbol $\approx$ means equality on the critical locus, i.e. \emph{on-shell}.

In the rest of the paper we show how to extend the correspondence between the symmetries also off shell in terms of the BV formalism.
\end{remark}

Observe that it is possible to add a cosmological term to BF theory as well by adding $\frac{\Lambda}{3}\int_{M} \mathrm{Tr} B^3$ to $S^0_{BF}$. In this case the theory admits an additional symmetry, namely $\delta_\Lambda A = \Lambda [B,\tau]\equiv \Lambda B\wedge\tau$. In the BV formalism this has to be complemented by additional higher terms for ghosts and antifields in order to yield a solution to the Classical Mater Equation (namely the term $c^\dagger \tau\wedge\tau$). Such solution is summarised with the introduction of \emph{superfields}, i.e. the inhomogeneous forms $\mathbb{B}=\tau + B +A^\dag + c^\dag $ and $\mathbb{A}= c + A + B^\dag + \tau ^\dag$. The BV-extended BF action in the superfield formalism then reads

\begin{equation}
S_{BF}(\Lambda) = \intl_M \mathrm{Tr}\left[\mathbb{B}F_{\mathbb{A}} + \frac{\Lambda}{3}\mathbb{B}^3\right]
\end{equation}

\section{Strong equivalence}
In this section we will prove that there is a strong equivalence between the BV theories $\mathfrak{F}_{BF}$ and $\mathfrak{F}_{GR}$. In order to do this we will adapt to three dimensions the strong equivalnce between the Palatini--Cartan BV formulation of gravity presented in \cite{CS2} and the version that was suggested by Piguet \cite{Piguet}. The main difference between the two BV theories is that the latter involves non-covariant expressions and non-global fields, but they are essentially equivalent up to (symplectic) field redefinitions.

\subsection{Non covariant BV teory}
Consider the assignment
\begin{equation}\label{Qpig}{
\begin{aligned}
&\mathrm{s}\,e'=L_{\xi'} e' + [\chi', e']\\
&\mathrm{s}\,\omega'= L_{\xi'} \omega' + d_{\omega'}\chi'\\
&\mathrm{s}\,\xi'=\frac{1}{2}[\xi',\xi']\\
&\mathrm{s}\,\chi'=L_{\xi'} \chi' + \frac{1}{2}[\chi',\chi']
\end{aligned}}
\end{equation}
defining a vector field $\mathrm{s}$ over 
$${F}_{PP}\coloneqq \Omega^1(M, \mathcal{W})\times\Omega^1(M,\Wedge{2}\mathcal{W})\times \mathfrak{X}[1](M) \times \Omega^0[1](M,\mathrm{ad}P)\ni(e',\omega',\xi',\chi')$$
which is cohomological, with $\xi$ a vector field with ghost number $\mathrm{gh}(\xi)=1$ and $\theta$ a function with values in $\Lambda^2 V$ and ghost number $\mathrm{gh}(\theta)=1$. The cotangent lift $\check{s}$ of $s$ to $\mathcal{F}_{PP}\coloneqq T^*[-1]{F}_{PP}$ is a cohomological vector field, that defines a BV-manifold together with the BV-extension of the Palatini action by $s$. We will denote such extension by 
\begin{multline}S_{PP}(\Lambda)=S^0_{GR}(\Lambda) + \intl_{M}\mathrm{Tr}\left\{e^\dag{}'\left(L_{\xi'} e' + [\chi', e']\right) + \omega^\dag{}'\left(L_{\xi'} \omega' + d_{\omega'}\chi'\right)\right\} + \\
+\mathrm{Tr}\left\{\chi^\dag{}'\left(L_{\xi'} \chi' + \frac{1}{2}[\chi',\chi']\right)\right\} + \frac12\iota_{[\chi',\chi']}\xi^\dag{}'.
\end{multline}
The subscript PP stands for Palatini-Piguet.

\begin{proposition}[\cite{CS2}]
The BV theory $\mathfrak{F}_{PP}=(\mathcal{F}_{PP},\Omega_{PP}, S_{PP},\check{s})$ is strongly equivalent to $\mathfrak{F}_{GR}$.
\end{proposition}

The proposition was proven for 4 space-time dimensions but it carries over in one less dimension without obstructions. If we select charts in $\mathcal{F}_{PP}$ and $\mathcal{F}_{GR}$ we can express the symplectomorphism in terms of the generating function\footnote{Observe that in \cite{CS2} the result differs by a trivial redefinition $c\mapsto -c$.}
\begin{equation}
F[\chi^\dag,\xi^\dag,e^\dag,\omega^\dag,e',\omega',\xi',\chi']\coloneqq \int_M \chi^\dag(\iota_{\xi} \omega' - \chi') + \iota_{\xi}\xi^\dag{}' - e^\dag e' - \omega^\dag \omega'
\end{equation}
and deduce the transformation rules as 
\begin{equation}\label{Hamiltongen}
p= - (-1)^{|q|}\pard{G}{q};\ \ \ \ Q =  (-1)^{|P|}\pard{G}{P},
\end{equation}
where $q=(\chi^\dag,\xi,e,\omega^\dag),\ P=(e^\dag{}',\omega',\xi^\dag{}',\chi')$.

We will see how the non covariant theory will turn out to be a simplifying intermediate step. In what follows we will indeed prove that there is the sequence of strong equivalences $\mathfrak{F}_{BF}\longrightarrow\mathfrak{F}_{PP}\longrightarrow\mathfrak{F}_{GR}$, that can be composed to yield the desired strong equivalence between Palatini gravity and BF theory in 3 dimensions.

\subsection{Strong Equivalence}
\begin{theorem}
The BV theories $\mathfrak{F}_{BF}$ and $\mathfrak{F}_{PP}$ are strongly equivalent. The symplectomorphism between their spaces of fields is given by the generating function
\begin{align}\label{GenFunPP}
G=& (-1)^{|\mathbb{P}|}\left[\mathbb{P} e^{-\iota_\xi}\mathbb{q}\right]^{top}\\\notag
=&-B^\dag\left(e - \iota_\xi\omega^\dag + \frac12\iota_\xi^2\chi^\dag\right) - \tau^\dag\left(-\iota_\xi e + \frac12\iota_\xi^2 \omega^\dag- \frac16\iota_\xi^3 \chi^\dag\right) - A\left(\omega^\dag - \iota_\xi \chi^\dag\right) - c\chi^\dag
\end{align}
where $\mathbb{P}=(B^\dag,\tau^\dag, A,c)$ and $\mathbb{q}=(e,\omega^\dag, \chi^\dag)$.
\end{theorem}

\begin{proof}
Observe that the set of \emph{old} coordinates is given by $q=(\xi,\mathbb{q})=(\xi,e,\omega^\dag,\chi^\dag)$, therefore, applying \eqref{Hamiltongen} we get the explicit transformation
\begin{subequations}\label{transf1}\begin{eqnarray}
B= e -\iota_\xi\omega^\dag + \frac12\iota_\xi^2\chi^\dag & B^\dag = e^\dag - \iota_\xi \tau^\dag\\
A=\omega - \iota_\xi e^\dag + \frac12\iota_\xi^2\tau^\dag & A^\dag = \omega^\dag - \iota_\xi\chi^\dag\\
c=\chi - \iota_\xi\omega + \frac12 \iota_\xi^2 e^\dag -\frac16\iota_\xi^3\tau^\dag & c^\dag = \chi^\dag
\end{eqnarray}\end{subequations}
together with
\begin{subequations}\label{transf2}\begin{align}
\tau= -\iota_\xi e + \frac12\iota_\xi^2\omega^\dag - \frac16\iota_\xi^3\chi^\dag \\ \label{implicitTAU}
\xi^\dag_a = \tau^\dag e_a + e^\dag \omega^\dag_a + \omega\chi^\dag_a
\end{align}\end{subequations}
Observe that Eq. \eqref{implicitTAU} can be solved in that $e$ is an isomorphism, but it will be easier to leave $\tau^\dag$ implicit when checking that the given transformation effectively yield the desired action-preserving symplectomorphism. Furthermore, notice that in computing $\xi^\dag_a=-(-1)^{|\xi^a|}\pard{G}{\xi^a}$we used $[\iota_\xi,\iota_{[\xi,\xi]}]=0$ several times.

Using the expression for $\xi^\dag_a$ we just found we can also compute $\iota_{\delta\xi}\xi^\dag = \iota_{\delta \xi} e \tau^\dag + \iota_{\delta\xi}\omega^\dag e^\dag + \iota_{\delta\xi}\chi^\dag \omega$, and $\iota_{[\xi,\xi]}\xi^\dag = \iota_{[\xi,\xi]}e \tau^\dag + \iota_{[\xi,\xi]}\omega^\dag e^\dag + \iota_{[\xi,\xi]}\chi^\dag \omega$ (notice the signs in commuting $\tau^\dag$ and $\iota_{\delta\xi} e$, as well as the sign in $\tau^\dag e_a \delta\xi^a = - \tau^\dag \iota_{\delta\xi}e$ owing to the fact that we consider $e_a$ as an odd object and $\iota_{\frac{\partial}{\partial x^a}}$ as an odd operator).

Then, if we denote the symplectomorphism generated in Eqts. \eqref{transf1} and \eqref{transf2} by $\phi\colon \mathcal{F}_{PP}\longrightarrow \mathcal{F}_{BF}$ we can verify with a long but straightforward computation that
\begin{equation}
\phi^*S_{BF}=S_{PP},
\end{equation}
concluding the proof.
\end{proof}

\begin{remark}
To verify that the symplectomorphism preserves the action functionals one might find a couple of identities particularly handy. In particular, from the fact that $\iota_{[\xi,\xi]} \iota_\xi \alpha = \iota_\xi\iota_{[\xi,\xi]} \alpha$, in the case $\alpha\in\Omega^{\text{top}}(M)$ and $\mathrm{dim}(M)=3$ we deduce 
$$\iota_\xi L_\xi^\omega \iota_\xi \alpha= - \frac13 d_\omega \iota_\xi^3 \alpha.$$
Similarly, from $\iota_\xi^2\iota_{[\xi,\xi]}\alpha = \iota_\xi\iota_{[\xi,\xi]}\iota_\xi\alpha= \iota_{[\xi,\xi]}\iota_\xi^2\alpha$, when $\alpha$ is a top-form we have
$$\iota_\xi^2d_\omega \iota_\xi^2 \alpha= \frac43\iota_\xi d_\omega\iota_\xi^3 \alpha.$$
Finally, one repeatedly needs to use the derivation property of the k-fold contraction with $\xi$ on $(\text{top} +k)$-forms, e.g. $0=\iota_\xi(\alpha\wedge\beta) = \iota_\xi\alpha\wedge\beta + \alpha\wedge \iota_\xi\beta$ and $\iota_\xi^2(\alpha\wedge \beta) = 2\iota_\xi\alpha\wedge\iota_\xi\beta + \iota_\xi^2\alpha \wedge\beta + \alpha\wedge\iota_\xi^2\beta$, implying $\iota_\xi^2\alpha \wedge\beta = \alpha\wedge\iota_\xi^2\beta$.
\end{remark}

Now we wold like to compose the two strong equivalences we have stated to construct a direct morphism between the covariant 3d Palatini BV-theory and BF-theory. In other words we want to complete the diagram
\begin{equation}
\xymatrix{
 \mathfrak{F}_{BF} \ar@{.>}[rr] \ar@{->}[dr]&&   \mathfrak{F}_{GR} \\
& \mathfrak{F}_{PP}\ar@{->}[ur] &
}
\end{equation}
This is clearly possible, and the next statement spells out the details.

\begin{theorem}
The BV theory for the 3-dimensional Palatini--Cartan theory of gravity is strongly equivalent to 3-dimensional BF theory.

The canonical transformation between the underlying spaces of fields is generated by the function
\begin{equation}
H=-B^\dag\left(e - \iota_\xi\omega^\dag -\frac12 \iota_\xi^2 \chi^\dag\right) - \tau^\dag\left(-\iota_\xi e +\frac12 \iota_\xi^2\omega^\dag + \frac13\iota_\xi^3 \chi^\dag\right) - A\omega^\dag - c\chi^\dag
\end{equation}
where fields are $(A,B,c,\tau)$ and their dual for BF-theory, and $(e,\xi,\omega,\chi)$ together with their duals for Palatini gravity.
\end{theorem}

\begin{proof}
To prove the result we use the standard formula for the composition of generating functions (for a recent account see, e.g., \cite[Section 3.3]{CDW} for the usual case and \cite[Section 2]{Anselmi} for the super case.). Consider two canonical transformations between three symplectic manifolds parametrised by, respectively, $(q_i,p_i)$. Denote the generating functions of such canonical transformations by $G(p_1,q_2)$ and $F(p_2,q_3)$, and assume that the generating functions are of degree $-1$ (this is the appropriate case for $-1$ shifted symplectic manifolds, and the procedure can be adapted to more general situations). Then we can find the generating function of the composed symplectomorphism as the critical point $(q_2,p_2)$ of the function\footnote{We will write $p_2q_2$ by simplicity, but indeed mean the sum over all components.}
$$\widetilde{H}(q_1,p_1,q_2,p_2,q_3,p_3) = G(p_1,q_2) + F(p_2,q_3) - (-1)^{|p_2|}p_2q_2$$
that is to say
$$H(p_1,q_1,p_3,q_3)\coloneqq\widetilde{H}\left(q_1,p_1,(-1)^{|p_2|}\frac{\partial\widetilde{H}}{\partial p_2},-(-1)^{|q_2|}\frac{\partial\widetilde{H}}{\partial q_2},q_3,p_3\right)$$
will generate the composite symplectomorphism.
\end{proof}

\subsection{A closer look to the cosmological term}
The pullback of the action functionals can be performed explicitly by applying the explicit formulae, and the equivalence can be verified directly. However, it is worth noting that the symplectomorphisms we found are stable under the addition of the respective cosmological terms, as follows from the following observations:

\begin{lemma}\label{cosmologicallemma}
Consider $\phi_{PP/GR}\colon \mathcal{F}_{PP/GR}\longrightarrow \mathcal{F}_{BF}$ and denote by $\mathbb{q}_i$ the i-th form in $\phi_{PP/GR}^*\left(e^{t\iota_\xi}\mathbb B\right)$. Then $\mathbb{q}_0=0$ and $\mathbb{q}_1=e$. Moreover, the function 
$$L(t)\coloneqq \intl_M \left(e^{t\iota_\xi}\mathbb B\right)^3$$
is constant. 
\end{lemma}
\begin{proof}
We write $e^{t\iota_\xi}\mathbb B = \sum_{i=0}^3\mathbb{B}^\xi_i$ and we can easily show that $\mathbb{B}^\xi_0=\tau + \iota_\xi B + \frac12 \iota_\xi^2 A^\dag + \frac16 \iota_\xi^3 c^\dag$, $\mathbb{B}^\xi_1 = B + \iota_\xi A^\dag + \frac12 \iota_\xi^2c^\dag$. Computing their pullback along the symplectomorphisms $\phi_{PP/GR}$ we obtain $\mathbb{q}_0=0$ and $\mathbb{q}_1=e$.

Finally, the time derivative $\dot{L}(t)$ reads
$$\dot{L}(t)=\intl_M \iota_\xi(e^{t\iota_\xi}\mathbb{B})^3 = 0$$
\end{proof}

\begin{proposition}
The symplectomorphisms $\phi_{PP/GR}$ map the cosmological term of the BV-extended BF theory to the cosmological term of General Relativity.
\end{proposition}

\begin{proof}
Since $L(t)$ is constant by Lemma \ref{cosmologicallemma}, we can replace the cosmological term $L(0)$ with $L(1)$ 
$$\intl_M \mathrm{Tr}[\Lambda \mathbb{B}^3] \leadsto \intl_M \mathrm{Tr}[\Lambda \left(e^{\iota_\xi}\mathbb B\right)^3]$$

Since by Lemma \ref{cosmologicallemma} $\mathbb{q}_0$ vanishes, it is easy to gather that the only term that contributes to $\phi^*_{PP/GR}\left(e^{t\iota_\xi}\mathbb B\right)^3$ is $\mathbb{q}_1^3=e^3$, proving the claim.
\end{proof}


\begin{thebibliography}{99}
\addcontentsline{toc}{chapter}{Bibliography}

\bibitem[Ans]{Anselmi} D. Anselmi, {\it Some reference formulas for the generating functions of canonical transformations}, Eur. Phys. Journal {\bf C 76}, (2016).

\bibitem[BV81]{BV81} I. A. Batalin and G. A. Vilkovisky. {\it Gauge algebra and quantization}, Phys. Lett. {\bf B 102}(1), 27-31 (1981).

\bibitem[BRST]{BRST} C. Becchi, A. Rouet and R. Stora, Phys. Lett. B52 (1974) 344.\\
C. Becchi, A. Rouet and R. Stora, Commun. Math. Phys. 42 (1975) 127.\\
C. Becchi, A. Rouet and R. Stora, Ann. Phys. 98, 2 (1976).\\
I. V. Tyutin, Lebedev Physics Institute preprint 39 (1975), arXiv:0812.0580.

\bibitem[Carlip]{Carl} S. Carlip, {\it Quantum Gravity in 2+1 Dimensions}, Cambridge Monographs on Mathematical Physics, Cambridge University Press (1998).
\bibitem[Car]{Cartan} E. Cartan. {\it Sur une gÈnÈralisation de la notion de courbure de Riemann et les espaces ‡ torsion}, C. R. Acad. Sci. {\bf 174}, 593?595 (1922).\\
E. Cartan, Comptes rendus hebdomadaires des sÈances de l'AcadÈmie des sciences,  174, 437-439, 593-595, 734-737, 857-860, 1104-1107 (January 1922)


\bibitem[CDW]{CDW} A. S. Cattaneo, B. Dherin, A. Weinstein, {\it Symplectic microgeometry II: generating functions}, Bull. Braz. Math. Soc., New Series  42, 507 (2011).

\bibitem[CMR14]{CMR1} A.S. Cattaneo, P. Mn\"ev, N. Reshetikhin, {\it Classical BV theories on manifolds with boundary}, Comm. Math. Phys. 332 (2): 535-603 (2014).

\bibitem[CS15]{CS1} A. S. Cattaneo and M. Schiavina, {\it BV-BFV approach to General Relativity: Einstein--Hilbert action}, J. Math. Phys. {\bf 57}(2) (2015).
\bibitem[CS17a]{CS4} A. S. Cattaneo and M. Schiavina, {\it The reduced phase space of Palatini--Cartan--Holst theory}, arXiv:1707.05351 (2017).
\bibitem[CS17b]{CS2} A. S. Cattaneo and M. Schiavina, {\it BV-BFV approach to General Relativity: Palatini--Cartan--Holst action}, arXiv:1707.06328 (2017).


\bibitem[Pal]{Palatini} A. Palatini, {\it Deduzione invariantiva delle equazioni gravitazionali dal principio di Hamilton}, Rend. Circ. Mat. Palermo {\bf 43}, 203 (1919).

[English translation by R.Hojman and C. Mukku in P.G. Bergmann and V. De Sabbata (eds.) Cosmology and Gravitation, Plenum Press, New York (1980)].

\bibitem[Piguet]{Piguet} O. Piguet, {\it Ghost Equations and Diffeomorphism Invariant Theories}, Class. Quant. Grav. {\bf 17}, 3799-3806 (2000).


\bibitem[Sta96]{Stash} J.Stasheff, {\it Deformation Theory and the Batalin-Vilkovisky Master Equation}, Deformation theory and symplectic geometry, proceedings, Meeting, Ascona, Switzerland, June 16-22, 1996, arXiv:q-alg/9702012.

\bibitem[Sta97]{stash97} J. Stasheff, {\it Homological reduction of constrained Poisson algebras}, J. Diff. Geom. {\bf 45}, 221-240 (1997).

\bibitem[Wit]{Witten} E. Witten, {\it 2+1 dimensional gravity as an exactly soluble system}, Nuclear Physics {\bf B} 311, 46-78 (1988/89).


\end{thebibliography}
\end{document}